\documentclass[letter, 10pt, conference]{ieeeconf}      

\IEEEoverridecommandlockouts                              

\usepackage{pgf}
\usepackage{listings}
\usepackage{algorithm} 
\usepackage{algpseudocode} 
\usepackage{fancyhdr}
\usepackage{flushend}
\usepackage{amsfonts,amssymb,amsmath,alltt}
\usepackage{stmaryrd}
\usepackage{xspace}
\usepackage{graphicx}
\usepackage{xcolor}
\def\BibTeX{{\rm B\kern-.05em{\sc i\kern-.025em b}\kern-.08em
    T\kern-.1667em\lower.7ex\hbox{E}\kern-.125emX}}

\usepackage{makeidx}
\usepackage[utf8]{inputenc}
\usepackage{url}
\usepackage[english]{babel}

\newenvironment{ttbox}{\begin{alltt}\ttbraces\small\tt}%
                      {\end{alltt}}
\def\ttbraces{\let\.=\nobreak\chardef\{=`\{\chardef\}=`\}\chardef\|=`\\}

\newcommand\ttand{\mbox{{$\land$}}}

\newcommand\ttcap{\mbox{{$\cap$}}}

\newcommand\ttfun{\mbox{{$\Rightarrow$}}}

\newcommand\ttimp{\mbox{{$\longrightarrow$}}}
\newcommand\ttequiv{\mbox{{$\equiv$}}}

\newcommand\ttforall{\mbox{{$\forall$}}}

\newcommand\ttin{\mbox{{$\in$}}}

\newcommand\ttImp{\mbox{{$\Longrightarrow$}}}

\newcommand\ttleq{\mbox{{$\le$}}}
\newcommand\ttrelI{\mbox{{$\to_{i}$}}}

\newcommand\ttrel[1]{\mbox{{$\to_{#1}$}}}
\newcommand\ttrelstar[1]{\mbox{{$\to_{#1}^*$}}}
\newcommand\ttalpha{\mbox{{$\alpha$}}}

\newcommand\ttsubseteq{\mbox{{$\subseteq$}}}

\newcommand\ttsO{\mbox{{\texttt{s}$_0$}}}
\newcommand\ttsone{\mbox{{\texttt{s}$_1$}}}

\newcommand\ttvdash{\mbox{{$\vdash$}}}

\newcommand\ttGn[1]{\mbox{$\texttt{G}_{\texttt{#1}}$}}
\newcommand\ttD[1]{\mbox{$\texttt{D}_{\texttt{#1}}$}}

\newcommand\ttF{\mbox{{${\sf{F}}$}}}

\newcommand\ttsum[1]{\mbox{{$\Sigma_{\texttt{#1}}$}}}
\newcommand\ttavg{\mbox{{$\triangledown$}}}
\newcommand\tteta{\mbox{{$\eta$}}}
\newcommand\ttlos{\mbox{$\mathcal{L}$}}
\newcommand\ttempty{\mbox{{$\varnothing$}}}
\newcommand\ttdisjone{\mbox{{$\stackrel{1}{\lozenge}$}}}
\newcommand\ttdisjx[1]{\mbox{{$\stackrel{#1}{\lozenge}$}}}
\newcommand\ttdisjunion[1]{\mbox{{$\stackrel{\bullet}{\bigcup}_{\texttt{#1}}$}}}
\newcommand\tteheps{\mbox{{$e^\epsilon$}}}
\newcommand\ttehneps{\mbox{{$e^{-\epsilon}$}}}
\newcommand\tteps{\mbox{{$\epsilon$}}}
\newcommand\Att{\mbox{{${\cal A}$}}}
\newcommand\rndass{\mbox{{$\mathord{\stackrel{\$}{\leftarrow}}$}}}
\newtheorem{theorem}{Theorem}
\newtheorem{definition}{Definition}[section]

\newtheorem{lemma}[definition]{Lemma}

\title{\LARGE \bf
  Formalisation of Security for Federated Learning with DP and Attacker Advantage
  in IIIf for Satellite Swarms -- Extended Version
}

\author{Florian Kamm\"uller$^{1,2}$ 
\thanks{$^{1}$Faculty of Science and Technology, Department of Computer Science,
        Middlesex University London, UK.
        {\tt\small <F.Kammueller@mdx.ac.uk>}}%
\thanks{$^{5}$Fakult\"at IV, Elektrotechnik und Informatik, TU Berlin, Germany.
{\tt\small <info@tu-berlin.de>}}%
}
\begin{document}

\maketitle
\thispagestyle{empty}
\pagestyle{empty}

\begin{abstract}
  In distributed applications, like swarms of satellites, machine learning can be efficiently applied even on
  small devices by using Federated Learning (FL). This allows to reduce the learning complexity by transmitting
  only updates to the general model in the server in the form of differences in stochastic gradient descent. FL
  naturally supports differential privacy but new attacks, so called Data Leakage from Gradient (DLG) have been
  discovered recently. There has been work on defenses against DLG
  but there is a lack of foundation and rigorous evaluation of their security. In the current work, we
  extend existing work on a formal notion of Differential Privacy for Federated Learning distributed dynamic
  systems and relate it to the notion of the attacker advantage. This formalisation is carried out within the
  Isabelle Insider and Infrastructure framework (IIIf) allowing the machine supported verification of theory and
  applications within the proof assistant Isabelle. Satellite swarm systems are used as a motivating use case but
  also as a validation case study.
\end{abstract}

\section{Introduction}
In this paper, we  review the formal characterization of Federated Learning (FL) and Differential Privacy (DP)
as well as its formalisation in the Isabelle Insider and Infrastructure framework (IIIf).
This method allows the logical characterization of FL as an inductive protocol definition in the IIIf 
providing a formal basis for expressing and verifying the security and privacy of applications. At the
same time, it provides a formal definition of DP over dynamic executions of the FL protocol. As the formalisation
is done in IIIf, the formalisation can be used to prove ``meta-level theory'', that is, theorems about the
FL-algorithm can be proved within the interactive Isabelle theorem prover. As an example, a proof of
a composition theorem is provided that allows to decompose the proof of DP of an FL architecture to the
clients' gradients. In addition to the meta-theory, it is possible to express application examples and reason
about their security by applying the provided generic meta-theory theory.

In this paper, we extend the FL-DP formalisation in IIIf by considering a more pragmatic way of characterizing
security: Attacker Advantage ($\text{Adv}(\Att)$). To this end, we formalise the notion of $\text{Adv}(\Att)$
and relate it to the one of FL-DP.
We further illustrate how the new meta-theory can be applied to a case study of satellite swarms.

\section{Background}

\subsection{Isabelle Infrastructure and Insider framework (IIIf)}
Isabelle is a generic Higher Order Logic (HOL) proof assistant. Its generic
aspect allows the embedding of so-called object-logics as new theories
on top of HOL. 
Object-logics, when added to Isabelle using constant and type definitions,
constitute a so-called {\it conservative extension}. This means that no 
inconsistencies can be introduced. Isabelle supports modular reasoning, that is,
we can prove theorems {\it in} an object logic but also {\it about} it. 
This allows the building of telescope-like structures in which a meta-theory
at a lower level embeds a more concrete ``application'' at a higher level.
Properties are proved at each level. Interactive proof is used to prove these 
properties but the meta-theory can be applied to immediately produce results.
Moreover, Isabelle can be used to build frameworks by providing generic meta-theory for
specific application domains.

The  Isabelle Infrastructure and Insider framework (IIIf) has been created initially for the modeling 
and analysis of Insider threats but is now a general framework for the state-based security
analysis of infrastructures with policies and actors.
Temporal logic and Kripke structure build the foundation. Meta-theoretical
results have been established to show equivalence between attack trees and CTL statements.
This foundation provides a generic notion of state transition on which attack trees and
temporal logic can be used to express properties. In addition, proofs are facilitated because
attacks can be derived using temporal logic analysis on Kripke structures by using the equivalence
between attack trees and CTL statements. Thereby, the verification corresponds to a complete
state exploration also known as model-checking.

The IIIF has been applied to various case studies.
In Section \ref{sec:concl}, we give as part of the related work a concise overview of the related applications of IIIf.

\subsection{Differential Privacy Federated Learning in IIIf}
The notion of Differential Privacy (DP) has become an ad hoc standard, most notably adopted for the US
census \cite{usc:20} to 
obfuscate personal data in public data records.
  Federated Learning (FL) \cite{mcm:17} is a widely used algorithm\footnote{Google scholar counts more than 60K
  citations for the seminal paper \cite{mcm:17}.} for machine learning in heterogeneous distributed applications.
  FL facilitates machine learning for mobile distributed scenarios. Its natural 
  relationship to DP has been investigated but the relationship to notions
  of noninterference in distributed systems is still largely open. 
  Federated Learning and Differential Privacy are naturally related but formally their relationship
  is concerned with partitioning and distributing data securely over a network (the DP part)
  while safely only returning the stochastic gradient descent (SGD) that ML has procured locally (the FL part).
  Both concepts have been studied excessively but not many (with a few notable exceptions, for example,
  \cite{tkd:11,lwfz:23,kpfn:24}) have observed that DP resembles the classical notion of noninterference used to
  express secure information flows.
  In this paper, we use this observation and extend it even further addressing the distribution and partitioning
  of the data over a network. We use the expressive power of the Isabelle Insider and Infrastructure framework
  (IIIf) to support this key issue of Federated Learning and relate them to classical Security notions of
  computational security by exploring the relationship to the cryptographic concept of Attacker Advantage ($\text{Adv}(\Att)$).

\subsection{DP and FL}  
The IIIf supports modeling infrastructures, human actors and policies. It also allows representing data
and its distribution across a network as well as the exchange of the SGDs between the client nodes
and the server. For modeling the notions of DP, an extension of the IIIf with probabilities \cite{kam:19c, kpfn:24}
supports quantified specification of security and privacy over infrastructures. Thus, DP can be
expressed directly as a probabilistic relation over a given infrastructure 
system. The use of IIIf allows expressing DP in a dynamic way as noninterference while taking into account
the nature of FL, that is, the (potentially unbalanced) distribution of data over a network. The novelty of
the approach \cite{kpfn:24} lies in considering DP and FL in the context of dynamically changing (data)
systems, that is, with an explicit representation of data distributed on clients and servers as well as
representing the FL algorithm formally as a protocol.
Compared to the vast amount on works exploring DP and FL on applications to data sets, the focus of
\cite{kpfn:24} is on the dynamicity (of data and communication) and data distribution aspects of DP and FL.
Also, by contrast, it formally models all concepts in IIIf thus introducing mathematical rigour and
automated verification.
  
The current paper now extends the existing NI-FL-DP concept by establishing -- within the formal framework
IIIf -- the conceptual relationship to classical complexity theory notions of computational security via the
Attacker Advantage ($\text{Adv}(\Att)$).

\subsection{Attacker Advantage}
  \label{sec:AA}
Many security notions are characterized by the {\it computational indistinguishability} between
two distributions $D_0$ and $D_1$. This computational indistinguishability is measured by the {\it advantage}
an adversary $\Att$ can have in distinguishing the two distributions. That is, we define the Advantage Adv($\Att$)
as follows.
\[
\text{Adv}(\Att) = P_{x\in D_1}(\Att(x) = 1) - P_{x \in D_0}(\Att(x) = 1)
\]
where $\Att(x) = 1$ means that the attacker $\Att$ guesses that the input $x$ is in distribution $D_1$
(while $\Att(x) = 0$ would mean $x \in D_0$). The notion $P_{x\in D_1}(\Att(x) = 1)$ is equivalent to
conditional probability $P(\Att(x) = 1 | x\in D_1)$.

For example, for the Diffie-Hellman key exchange protocol, with Diffie-Hellman tuples in $\mathbb{G} = \langle g\rangle$,
a group of prime order $p$, spanned by a generator $g$, and denoted additively.
\begin{definition}[Decisional Diffie-Hellman Problem (DDH)]
  \label{def:ddh}
  The DDH assumption in a group $\mathbb{G}$ 
  of prime order $p$, with a generator $g$, states that the distributions $D_0$ and $D_1$ are computationally
  hard to distinguish, where
\[
\begin{array}{ccc}
   D_0 & = & \{ (a\cdot g,b\cdot g,ab\cdot g). a,b \rndass Z_p\} \\
   D_1 & = & \{(a\cdot g,b\cdot g,c\cdot g). a,b,c \rndass Z_p\}
\end{array}
\]
\end{definition}
The notation $k \rndass K$ is used in Cryptography to express ``selecting a random string from
[a uniformly distributed set of keys] $K$ and naming it $k$'' \cite{gb:08}[p. 63]. The set of keys
in our case is $Z_p$ which is a another common mathematical notion for the prime number field
$\{0, \dots, p-1\}$ with addition and multiplication modulo $p$.
We then denote as Adv$^{\texttt{ddh}}_G({\cal A})$ the advantage of an adversary $\Att$ in distinguishing
$D_0$ and $D_1$. Looking back at the definition of Adversarial Advantage above, this means that if the DDH-problem
is indeed computationally hard (as we believe it to be), then the attacker will not be able to distinguish whether
$c$ is the product $ab$ of the secrets $a$ and $b$. That is, they will output a 1 with about the same
probability when exposed to an example that is in $D_0$ as for one that is in $D_1$, so that the difference
$\text{Adv}(\Att)$ (defined as $\text{Adv}(\Att) = P_{x\in D_1}(\Att(x) = 1) - P_{x \in D_0}(\Att(x) = 1)$) is close to 0.

This example illustrates the notion of $\text{Adv}(\Att)$ on the well-known DH problem; we will adapt it to the setting of
FL and DP.

To this end, we need some more tools.
A global assumption is that distinguishing distributions is symmetric. That is, 
  $P_{x \in D_0}(\Att(x = 0)) = P_{x \in D_1}(\Att(x = 1))$.

\begin{lemma}\label{lem:adv}
  \[
  \text{Adv}(\Att) = 2 \times P_{\substack{x\in D_b\\b\in\{0,1\}}}(\Att(x) = b) - 1
\]
\end{lemma}
\begin{proof}
From the definition of advantage we have
  \[
\text{Adv}(\Att) = P_{x\in D_1}(\Att(x) = 1) - P_{x \in D_0}(\Att(x) = 1)
\]
By Bayes rule $P(A|B) = 1 - P(\overline{A}|B)$ this is equal to
\[ P_{x\in D_1}(\Att(x) = 1) - (1 - P_{x \in D_0}(\Att(x) = 0)) \]
which is the same as
\[ P_{x\in D_1}(\Att(x) = 1) - 1 + P_{x \in D_0}(\Att(x) = 0) \]
Replacing the symmetry assumption above on the right, we get
\[ P_{x\in D_1}(\Att(x) = 1) - 1 + P_{x \in D_1}(\Att(x) = 1) \]
which rearranges to
\[ 2 \times P_{x\in D_1}(\Att(x) = 1) - 1\]
Replacing on the left instead yields
\[ P_{x\in D_0}(\Att(x) = 0) - 1 + P_{x \in D_0}(\Att(x) = 0) \]
thus rearranging to
\[ 2 \times P_{x\in D_0}(\Att(x) = 0) - 1\]
Generalizing these two conclusions, we get
\[2 \times P_{\substack{x\in D_b\\b\in\{0,1\}}}(\Att(x) = b) - 1 \]
as we needed to show. \end{proof}

\section{Formalising the Attacker Advantage for FL-DP}

\subsection{FL Algorithm}
FL combines the requirements of training on distributed clients where the data
may be unevenly distributed (``unbalanced'' \cite{mcm:17}) and local data sets may not be representative
of the global one (``non-iid'' \cite{mcm:17}), massively distributed and with limited communication.

\begin{algorithm}
	\caption{FL} 
	\begin{algorithmic}[1]
          \State {\bf Server executes}:
             \State initialize $w_0$; $K$ := set of clients
	     \State \qquad {\bf for} each round {$t=1,2,\ldots$}
             \State \qquad \qquad {\bf for} each client {$k \in K$} {\bf in parallel do}
	     \State \qquad \qquad \qquad  $w_{t+1}^k :=$ ClientUpdate($k$,$w_t$)
             \State \qquad \qquad \qquad $w_{t+1} := \sum_{k+1}^K \frac{n_k}{n} w_{t+1}^k$
             \State {\bf ClientUpdate}($k$,$w$):
              \State  \qquad $w := w - \eta \ttavg \ttlos(w;P_k)$
              \State \qquad return $w$ to Server
	\end{algorithmic} 
\end{algorithm}

The FL algorithm can be applied to any finite-sum objective as follows \cite{mcm:17}
\[
  \min_{w \in \mathbb{R}^d} f(w) \qquad \text{where}   \qquad f(w) \stackrel{\text{def}}{=} \frac{1}{n} \sum_{i = 1}^n f_i(w).
\]
where the $f_i$ are loss functions, that is, $f_i(w) = \ttlos(x_i,y_i;w)$ where $(x_i,y_i)$ is an example
data and $w$ is the model parameter, that is, the current value for the parameters of the ML algorithm
to get the best prediction. So, the FL algorithm tries to approximate a prediction function by minimizing
a loss function (sometimes called cost function) on data example points (the training data).
What is more is that the FL algorithm considers a set of clients $K$ over which the data is partitioned.
Let $P_k$ be the partition for client $k \in K$ and $n_k = |P_k|$ the number of data points in that partition.
This adds a layer of detail to highlight the unbalanced scenario for the FL algorithm. We can re-write
the above description by detailing the averaging in each client.
\[
f(w) = \sum_{k=1}^K \frac{n_k}{n}F_k(w) \qquad \text{where} \qquad F_k(w) = \frac{1}{n_k} \sum_{i\in P_k} f_i{w} .
\]
\subsection{Formalisation of FL-DP protocol}
The formalization of FL in IIIf implements the FL algorithm introduced in the previous sections as 
protocol. 
That is, the execution of the FL algorithm is a sequence of exchanges between the server and a number
of distributed clients. The framework IIIf supports well this formalization of protocols. For example,
Quantum Key Distribution \cite{kam:19c,knpw:25} also models protocols built on top of the basic Kripke
structures with CTL and Attack Trees thus providing concepts and techniques for formalizing FL in IIIf.

In general, infrastructures in the IIIf are graphs with actors and local policies attached to the nodes.
For the FL algorithm, the actors are the server and its clients attached each to their nodes (\texttt{location}).
The data can be explicitly represented in infrastructures and is also attached to the nodes corresponding to the
server and its clients. 
The constructor \texttt{Lgraph} of the data type \texttt{igraph} assembles the components
\texttt{gra}, \texttt{aloc}, \texttt{server}, \texttt{clients}, \texttt{ready}, \texttt{gradient},
\texttt{curmodpar}, \texttt{partition},
and \texttt{dataset} which represent: the actual topology of locations, the alocation of actors (that is, also
devices) to the locations, the server and clients, the clients finished (``ready'') with learning in a round of the
FL algorithm, the local gradients currently held in a location, the global current model
parameter used for the distributed learning, the partition of the global data set onto the distributed system, and the
global data set. The name introduced for each component can be used as a projection functions on the \texttt{igraph} to
select each of the components. The parameter type \texttt{\ttalpha} in the components \texttt{gradient} and
\texttt{curmodpar} is given as a type variable. This permits representing arbitrary vectors of model parameters
(coefficients).
\begin{ttbox}
{\bf{datatype}} igraph = Lgraph{\bf gra:} location set
                     {\bf aloc:} actor \ttfun location
                     {\bf server:} actor
                     {\bf clients:} actor set
                     {\bf ready:} actor set
                     {\bf gradient:} actor \ttfun {\ttalpha :: num}
                     {\bf curmodpar:} \ttalpha
                     {\bf partition:} actor \ttfun data set
                     {\bf dataset:} data set                           
\end{ttbox}
The type \texttt{igraph} represents the infrastructures that we use for modeling distributed scenario.
Protocol execution traces are sets of lists of events: each list represents a finite execution
trace of the protocol. The set of such traces is a snapshot of execution traces of the system.
The possibility of consideration of parallel traces is crucial for security; many attacks come from
exploits between simultaneous protocol runs.
\begin{ttbox}
{\bf{datatype}} protocol = Protocol {\bf evs:} event list set
\end{ttbox}
The state is given by a datatype called \texttt{infrastructure} containing the \texttt{igraph} together with
the protocol and a set of local policies attached to graph locations.
Infrastructures contain the topology of the network with the current data distribution as well as a current
set of protocol traces and the local policies to distribute a global security policy. The policy component
\texttt{\bf poli} is relevant for applications to define local data and action policies.
\begin{ttbox}
{\bf{datatype}} infrastructure =
Infrastructure {\bf igra:} igraph 
               {\bf prot:} Protocol
               {\bf poli: }[igraph, location] \ttfun policy set
\end{ttbox}
The semantics of a distributed system is formalized as state transition relations over this infrastructure type.
This can be easily done in the IIIf by instantiation of the generic Kripke structures. Thereby automatically
providing the meta-theory of temporal logic, attack trees and model checking. 

For FL the events in a protocol trace correspond to the communication actions of the FL algorithm:
as an initial step {\it putting} the partition of the data on the clients, {\it getting} back the local
gradient descents from the clients to the server, and {\it evaluation} of the combined local gradient
descents by the server once all clients have been harvested to compute and redistribute the new current
model parameter for the next learning round.

The theory of Kripke structures and temporal logics CTL in the IIIf supports a generic state
transition relation that can be instantiated as \texttt{\ttrel{FL}} to implement the FL algorithm.
This relation can then be defined as an inductive relation over the infrastructure type.
Consequently, it has three cases:
\begin{itemize}
\item \texttt{\bf put\_part} is the initial deployment of the partitioned global data set to all the
  clients 
  by inserting the data into disjoint partitions.  
\item \texttt{\bf get\_grad} is the case where a local gradient results from the local
  learning in a client the position \texttt{c}. 
  The result is loaded into the graphs' component \texttt{gradient G} by updating 
  this function for client \texttt{c}.
  Also, the current client \texttt{c} is added to the set \texttt{ready G} signifying that \texttt{c} has
  updated its gradient in this round of FL.
\item \texttt{\bf eval\_server} finally, pulls together the new model parameter 
  by averaging all local gradients. 
  The event \texttt{\bf eval\_server} is only enabled if the set \texttt{ready} contains all clients. After
  the eval step, \texttt{ready} is initialized to the empty set for a new round.
\end{itemize}  
For the inductive definition of the state transition system of \texttt{\ttrel{FL}} that defines the above
informally described semantics of the FL algorithm in IIIf. 
\begin{ttbox}
{\bf{inductive}} state_transition_FL ::
  [infrastructure, infrastructure] \ttfun bool
                              ("_ \ttrel{FL} _")
{\bf{where}} 
{\bf{put\_part}}:
 G = igra I \ttImp P = prot I \ttImp [] \ttin evs P \ttImp
 dataset G = \ttdisjunion{c \ttin clients G} (Part c) \ttImp
 P' = insertp [Put [Part]] (evs P) \ttImp
 G' = Lgraph (gra G)(aloc G)(server G)(clients G)
      (gradient G)(curmodpar) Part (dataset G) \ttImp
 I \ttrel{FL} Infrastructure G' P' (poli I) |
{\bf{get\_grad:}}
 G = igra I \ttImp P = prot I \ttImp l \ttin evs P \ttImp
 c \ttin clients G \ttImp
 Gradc =
 (curmodpar G) - \tteta * \ttavg \ttlos(curmodpar G, partition c)
 P' = insertp (Get c Gradc # l) (evs P) \ttImp
 G' = Lgraph (gra G)(aloc G)(server G)(clients G)
           (insert c (ready G))
           (gradient G(c := Gradc))
           (curmodpar G)(partition G)(dataset G) \ttImp
 I \ttrel{FL} Infrastructure G' P' (poli I) |
{\bf{eval\_server:}}
 G = igra I \ttImp P = prot I \ttImp l \ttin evs G \ttImp
 ready G = clients G \ttImp n = card(dataset G)
 Newmodel =
 \ttsum{c \ttin clients G} (card (partition G c) / n) * gradient G c
 P' = insertp (Eval Newmodel # l) (evs P) \ttImp
 G' = Lgraph (gra G)(aloc G)(server G)(clients G)
           \ttempty
           (gradient G)
            Newmodel
           (partition G)(dataset G) \ttImp
 I \ttrel{FL} Infrastructure G' P' (poli I)                
\end{ttbox}

When formalizing the security properties of FL systems we need to express probabilities of state transition
execution traces. To this end, the IIIf offers a probabilistic path quantifier \cite{kam:19c} that allows to
assign probabilities to paths that end in a specified set of states. Thus, specific properties can be expressed
and quantified by probabilities, for example, that an FL system leads to states with certain model parameters
or in which specific gradients were transmitted.
This path quantifier is relative to a Kripke structure \texttt{M} because ``gradient transmission''
paths must start from initial states and be paths along states within \texttt{M}. The IIIf provides the
following specialized path selector \texttt{M \ttvdash\ttF\;\! s} for Kripke model \texttt{M} and state  \texttt{s}.
The list operator \texttt{hd} yields the head (first element) of a list -- here ensuring the trace \texttt{l} starts
in an initial state -- and the operator \texttt{F} selects a set of finite paths leading into a state \texttt{s}.
\begin{ttbox}
{\bf{definition}} M \ttvdash\ttF s \ttequiv
   \{ l. set l \ttsubseteq states M \ttand hd l \ttin init M \} \ttcap F s
\end{ttbox}
This is now combined with the probability notions provided in the IIIf framework. For example, to introduce the 
probability for a set of finite paths given an assignment \texttt{ops} of probabilities to  paths 
representing outcomes of protocol runs, the operation \texttt{pmap ops} produces a probability distribution
for this protocol. This can then be applied to the previous path selection:
\texttt{pmap ops (M \ttvdash\ttF s)} computes the probability \texttt{P(M \ttvdash\ttF s)} of the event that
paths end in state \texttt{s}.

Given these notions, probabilistic noninterference for FL can be defined based on DP.

\subsection{Definition of NI-FL-DP and Compositionality Theorem}
\label{sec:nifldp}
Noninterference (NI) \cite{gm:82} is a formal definition of secure information flows.
It is an equivalence property over sequences of state transitions (or sequences of events --
depending on the system view).   
Intuitively, it characterizes that from a certain viewpoint various executions of a system do not
reveal any secret information.
For programming languages, noninterference is defined as an indistinguishability relation of
program executions. The indistinguishability is respective to visible system states or visible
input/output variables. 
It thus requires some observer viewpoint level which is commonly added as H(igh) and L(ow) labels that
are variables. It then declares system execution (states) to be equivalent if they only differ
in values of H-variables. Noninterference is often defined as a Low-bisimulation:
in any program run the indistinguishability of pairs of system states is preserved. This implies that
the computation of L-variables is independent of that of H-variables. No unintended ``illicit'' information
flows from H to L can appear in program execution.

The notion of NI of FL by DP (NI-FL-DP) \cite{kpfn:24} is inspired by Volpano and Smith's
probabilistic NI \cite{vs:99}: ``under two memories that may differ
[only] on H-variables, we want to know that the probability that a concurrent program can reach some global
configuration under one of the memories is the same as the probability that it [the concurrent program] reaches
an equivalent configuration under the other''  \cite{vs:99}.
For ``concurrent program'', we read ``distributed infrastructure'', for ``differs on H-variables'' we have
``the datasets differ only by one'' and for ``equivalent configuration'' we have differential privacy, that is,
the probability of reaching equivalent configurations differs only by the DP-factor.
Thus, this notion gives rise to a notion of NI-FL-DP: noninterference of federated learning as differential
privacy,

The definition of NI-FL-DP in IIIf states that for two infrastructure states 
\texttt{\ttsO} and \texttt{\ttsone} that both originate in one initial infrastructure \texttt{I} and whose
data sets differ only by one element (expressed using the \texttt{\ttdisjone} operator), the probabilities of
finite traces leading into such states differ only by the DP-factor \texttt{\tteheps}.
\begin{ttbox}
{\bf{definition}} NI\_FL\_DP I \ttsO \ttsone \tteps \ttequiv 
 I \ttrelstar{FL} \ttsO \ttImp I \ttrelstar{FL} \ttsone \ttImp
 M = \{Kripke \{t . I \ttrelstar{FL} t\}\{I\}\} \ttImp
 \ttGn{0} = igra \ttsO \ttImp \ttGn{1} = igra \ttsone \ttImp
 dataset \ttGn{0} \ttdisjone dataset \ttGn{1} \ttImp 
 P (M \ttvdash\ttF \{s. curmodpar(igra s) = curmodpar \ttGn{0}\}) \ttleq 
 \tteheps * P(M \ttvdash\ttF \{s. curmodpar(igra s) = curmodpar \ttGn{1}\})
\end{ttbox}
Note, that we can safely add $\epsilon \geq 0$ as assumption to NI-FL-DP
because $\ttdisjone$ is symmetric, that is
\begin{ttbox}
  A \ttdisjone B \ttequiv B \ttdisjone A
\end{ttbox}  
If for a given $\epsilon$, we would have $\epsilon < 0$, can just swap the order of the data sets because
\begin{ttbox}
  \ttD{0} \ttleq \tteheps \ttD{1} \ttequiv \ttD{1} \ttleq \ttehneps \ttD{0} 
\end{ttbox}  

Since the notion NI-FL-DP is inspired by probabilistic NI for distributed systems, it is not surprising that the
methodologies and techniques applied in the NI domain can be adapted to provide useful tools for the verification
and analysis also for NI-FL-DP. Many works on NI (and already Goguen and Meseguer in their seminal paper)
introduced so-called ``unwinding relations'' to simplify proving noninterference \cite{gm:84}. Unwinding theorems
allow stepwise (sequential) decomposition of NI. Since the NI-FL-DP property is a security
notion for {\it distributed} systems, an unwinding theorem for this notion consequently needs to decompose
the global property to its (distributed) components. The following compositionality theorem is provided in the
IIIf framework to break down the proof of security of an FL system to that of its client components.
\begin{ttbox}
{\bf{theorem}} NI\_Decomposition:
 I \ttrelstar{FL} \ttsO \ttImp I \ttrelstar{FL} \ttsone \ttImp
 M = \{Kripke \{t . I \ttrelstar{FL} t\}\{I\}\} \ttImp
 \ttGn{0} = igra \ttsO \ttImp \ttGn{1} = igra \ttsone \ttImp
 clients \ttGn{0} = clients \ttGn{1} \ttImp 
 \ttforall c \ttin clients \ttGn{0}. partition \ttGn{0} c \ttdisjone partition \ttGn{1} c
 \ttimp
 P(M \ttvdash\ttF \{s. gradient(igra s) c = gradient \ttGn{0} c\}) \ttleq 
 \tteheps * P(M \ttvdash\ttF \{s. gradient(igra s) c = gradient \ttGn{1} c\})
 \ttImp NI\_FL\_DP I \ttsO \ttsone \tteps
\end{ttbox}  
The precondition for this Decomposition Theorem assumes that
for two state graphs that differ in their local dataset by only one item of the data set,
each client's gradient is differentially private in the local learning of the FL-algorithm.
Given this precondition, the theorem allows to infer that the global algorithm  computing
the average of the local gradients, is also differentially private. That is, the \texttt{NI\_FL\_DP}
property holds.

With the notion of NI-FL-DP, we have a good method for expressing information flows -- and thus leakage --
for FL. In fact, it is very desirable to have quantifiable notions of information flow because in reality
we cannot hope to have 100\% security but at least need a way to quantify it.
The additional Decomposition Theorem is furthermore a very practical tool since it allows to break down the
global NI-FL-DP property to that of the local components of the FL system. Practically, the $\epsilon$ can
be computed as the maximum value of the local components' $\epsilon$.
Ultimately, the $\epsilon$ values are, however, measured values providing quantification of an upper bound
based on empirical data.
To counter this rather empirical way of estimation a security and privacy boundary, it is desirable to have a
more classical metric. This motivates the following extension of relating the NI-FL-DP property to $\text{Adv}(\Att)$
(see Section \ref{sec:AA}).

\subsection{Relationship to Attacker Advantage}

\subsubsection*{Decisional NI-FL-DP problem}
The assumption NI-FL-DP can be expressed as a decision problem following the outline of Definition
\ref{def:ddh}.
\begin{definition}[Decisional NI-FL-DP Problem]\label{def:nifldpprob}
  Given that
  \begin{ttbox}
 I \ttrelstar{FL} \ttsO \ttand I \ttrelstar{FL} \ttsone \ttand
 M = \{Kripke \{t . I \ttrelI^* t\}\{I\} \} \ttand
 \ttGn{0} = igra \ttsO \ttand \ttGn{1} = igra \ttsone \ttand
 dataset \ttGn{0} \ttdisjone dataset \ttGn{1} 
  \end{ttbox}  
  the NI-FL-DP assumption states that the
  distributions $\ttD{0}$ and $\ttD{1}$ are computationally hard to distinguish, where
\begin{ttbox}
  \ttD{0} = M \ttvdash\ttF \{ s. curmodpar(igra s) = curmodpar \ttGn{0} \}
  \ttD{1} = M \ttvdash\ttF \{ s. curmodpar(igra s) = curmodpar \ttGn{1} \}
\end{ttbox}
We denote as Adv$^{\text{dnifldp}}(\Att)$ the advantage of an adversary $\Att$ in distinguishing $\ttD{0}$ and $\ttD{1}$
when seeing (public) current model parameters \texttt{curmodpar \ttGn{$i$}}, $i\in\{0,1\}$.  
\end{definition}
Note, that the above distributions are sets of traces of the state transition relation of the
FL system leading into the state sets designated after the \texttt{\ttvdash\ttF}-operator.
Thus the distributions $\ttD{0}$ and $\ttD{1}$ are effectively the same probability distributions as
used in the NI-FL-DP definition.
Also it is important to note at this point that the NI-FL-DP problem is defined as an assumption.
Clearly, the computational hardness of the indistinguishability depends on the problem of inverting the
prediction function based on the gradients.
However, Definition \ref{def:nifldpprob} 
provides new metrics for security and privacy of FL systems based on the concept of DP. 

The definition of the assumptions can be applied to categorize the security of techniques for defending
FL architectures using the DP techniques. For example, Chang and Zhu \cite{wz:24} report on the concept
of ``Deep Leakage from Gradient (DLG)'' \cite{zhu:19} which refers generally to clients inadvertently revealing
sensitive properties of their local training data. 
In \cite{wz:24}, the authors report on two techniques, {\it gradient sparsification} and {\it pseudo-gradients}, to protect
from DLG. The contribution of our paper is to make the security and privacy of such protection mechanisms quantifiable and
amenable to automated reasoning and verification in IIIf.
To this end the definitions presented in this section can be applied. 

We can use here the notion of $\text{Adv}(\Att)$ (see Section \ref{sec:AA}) to effectively relate the computational
security characterization to DP. 
The following theorem uses the equivalence between $\text{Adv}(\Att)$ and the DP-factor for FL to quantify the
ability of $\Att$ to identify whether an element is in a dataset from the gradinet outputs.
\begin{theorem}[Adv-FL-DP-Eq]
  \label{thm:advfldpeq}
  Let Kripke structure \texttt{M} be defined as
  \texttt{\{Kripke \{t. I \ttrelstar{FL} t\}\{I\}\}}.
  Furthermore, let
  \texttt{I \ttrelstar{FL} \ttsO} and \texttt{I \ttrelstar{FL} \ttsone} where
  \texttt{\ttGn{0} = igra \ttsO} and \texttt{\ttGn{1} = igra \ttsone} such that
  \texttt{dataset \ttGn{0} \ttdisjx{x}\ dataset \ttGn{1}},
  that is, their datasets differ only in $x$.  
  Given two distributions
  \begin{ttbox}
  \ttD{0} = M \ttvdash\ttF \{ s. curmodpar(igra s) = curmodpar \ttGn{0} \}
  \ttD{1} = M \ttvdash\ttF \{ s. curmodpar(igra s) = curmodpar \ttGn{1} \}
  \end{ttbox}
  Then $P_{\substack{x\in D_b\\b\in\{0,1\}}}(\Att(x) = b) = 1 - \frac{1}{2e^{-1}}$  where $e^\epsilon = \frac{P(\ttD{0})}{P(\ttD{1})}$
  and $D_b =$ \texttt{dataset \ttGn{b}}.
\end{theorem}
\begin{proof}
Intuitively, the distributions $\ttD{0}$, $\ttD{1}$ are computationally hard to distinguish iff
their NI-FL-DP factor $\epsilon$ is close to 0.
By Definition \ref{def:nifldpprob}, the attacker's advantage and \texttt{NI\_FL\_DP \ttsO\, \ttsone\, \tteps} have
the following equivalences
     \[
   \begin{array}{ccc}
  \text{Adv}^{\text{dnifldp}}(\Att) = 0        & \equiv & \frac{1}{e^\epsilon} = 1,\ \text{if\ } \epsilon = 0\\
  0 < \text{Adv}^{\text{dnifldp}}(\Att) \leq 1 & \equiv & \frac{1}{e^\epsilon} \to 0,\ \text{if}\ \epsilon > 0\\
   \end{array}
   \]
   For the case $\epsilon = 0$, we have perfect indistinguishability, i.e. \texttt{NI\_FL\_DP \ttsO\, \ttsone\, 0},
   and the attacker's advantage is also 0. 
For the non-zero case, the advantage grows towards 1 as $\epsilon$ goes to infinity.
Summarizing, we have  $\text{Adv}^{\text{dnifldp}}(\Att) = 1 - \frac{1}{e^\epsilon}$.
Unfolding the definition of advantage (see Section \ref{sec:AA}) and using Lemma \ref{lem:adv} we get
\[ 2 \times P_{\substack{x\in D_b\\b\in\{0,1\}}}(\Att(x) = b) - 1  = 1 - \frac{1}{e^\epsilon} \]
Transforming the equation, we immediately arrive at
\[ P_{\substack{x\in D_b\\b\in\{0,1\}}}(\Att(x) = b) = 1 - \frac{1}{2e^{-1}}\]
as requested.
\end{proof}
Theorem \ref{thm:advfldpeq} can also be decomposed by applying Theorem \texttt{NI\_Decomposition} to
break down the Advantage to the client components.

\section{Satellite Swarm Application}
Satellite swarms are collaborations of satellites that function as distributed systems in space.
Their spatial distribution allows them to operate from different locations -- often for observational purposes
-- and share tasks to reduce the operational load.
An example is the European Space Agency's (ESA) project Swarm (see Figure \ref{fig:swarm}).
\begin{figure}
\vspace{-.5cm}
  \begin{center}
  \includegraphics[scale=.41]{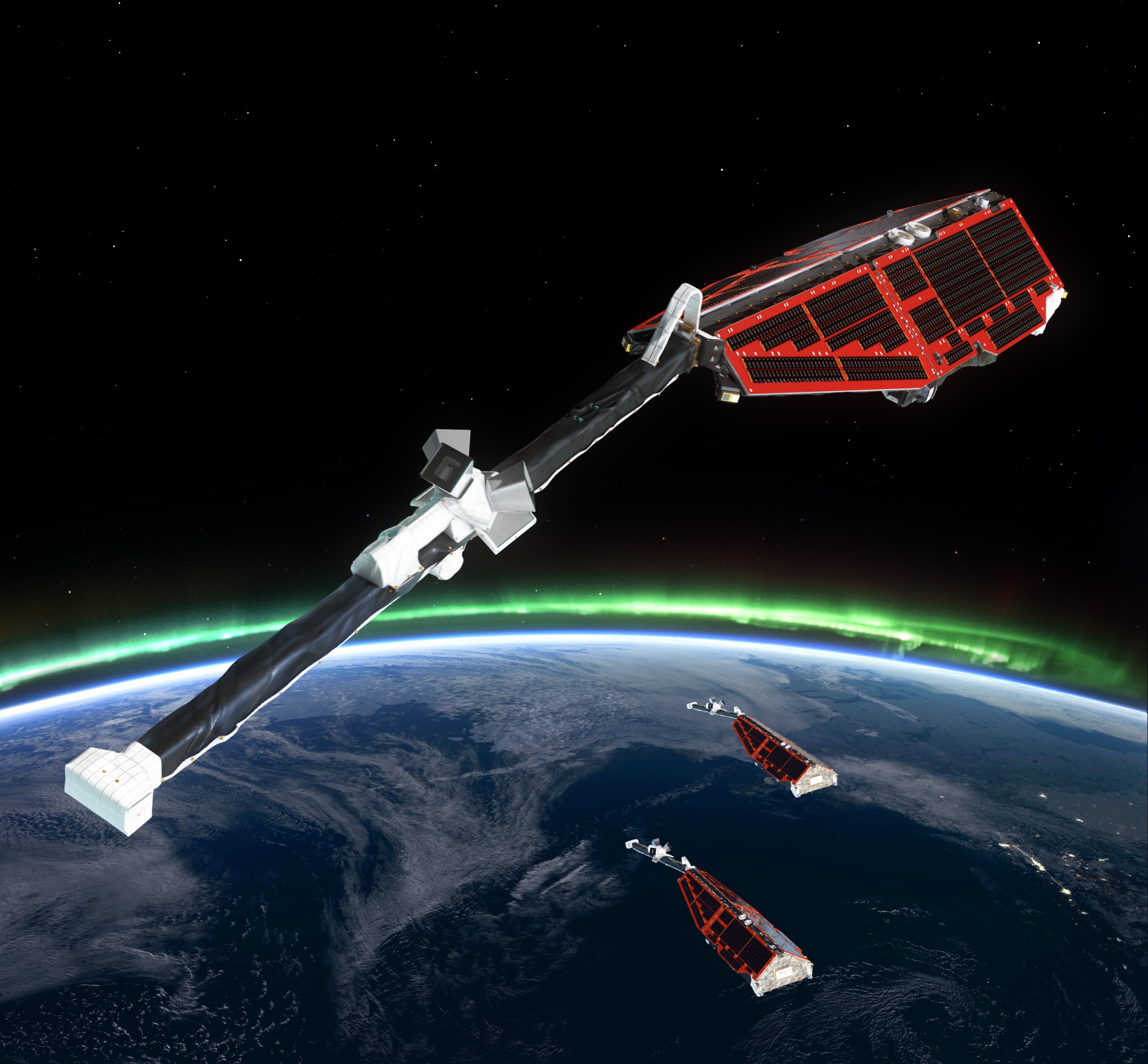}
\end{center}
\vspace{-.5cm}
\caption{Swarm is ESA's first constellation of Earth observation satellites designed to measure the magnetic signals from Earth’s core, mantle, crust, oceans, ionosphere and magnetosphere, providing data that will allow scientists to study the complexities of our protective magnetic field, \copyright{ESA}.}
\label{fig:swarm}
\vspace{-.5cm}
\end{figure}
Scientific projects like Swarm usually have a humanitarian and ethical motivation and do not
seek to hide their data. However, it can be easily imagined that security and privacy can become an issue,
for example, to protect private data of individuals or security criticality, like locations of secret military
assets.
As an illustrative case study, let {\it Moniteo} be a Swarm-like application to monitor temperature data on the Earth's
surface using FL. That is, each Moniteo's agent satellites $\sigma$ is geo-stationary and equipped with an infrared camera.
The mission is to find a hypothesis (model) for the temperature in given locations using FL.
Initially, each $\sigma$ is equipped with a dataset $\text{Part}_\sigma$ of locations and their temperatures.
Certain locations must be omitted in the data set because they convey private information. 
That is, for geographic data sets $D_0, D_1$ we have that $D_0 \ttdisjone D_1$ for the partition of data sets.
Even if $\sigma$ learns no data, it will only serve to improve the local hypothesis: the new data does not need to
be transported to the server.
Now, Moniteo's security depends on low (or no) information flows from the local data to the SGDs transmitted to the
server. Here, Theorem \ref{thm:advfldpeq} and its corollary for the decomposition of the Advantage yield that local
deviations of data, e.g., by severe temperature increase in a new data point, cannot be distinguished by an attacker.
That is, if a client satellite $\sigma$ would have two data partitions that differ in one data point because of a
different measurement of that one data point, then the difference in the gradients would result in a difference
below $e^\epsilon$. Consequently, if the satellite $\sigma$'s infrared cameras measure a significant difference
of temperature in a location that is marked as secret, it may safely omit the data point's value from the
gradient computation without risking significant information flows above the threshold of computational security.

\section{Conclusions and Related Work}
\label{sec:concl}
In this paper, we have revisited the notion of Information Flow Security (NI) for Federated Learning (FL) by
Differential Privacy (DP) (NI-FL-DP) and related it to Computational Security given by the cryptographic notion
of Attacker Agvantage Adv($\Att$). We illustrated its importance by a satellite swarm use case.

We review here some relevant related works.
\subsection{Related Work on the Isabelle Insider and Infrastructure framework}
\label{sec:relisa}
A whole range of publications have documented the development of the Isabelle Insider framework.
The publications \cite{kp:13,kp:14,kp:16} were first to apply formal verification to insider threats
expressing policies, attacks and insider behaviour in Modelcheckers and Isabelle.
As a generic validation it has been shown that the insider threat patterns defined in the CERT guide
on insider threats \cite{cmt:12} could be handled. 
After this initiation of an  Isabelle Insider framework, it has been applied to auction protocols
\cite{kkp:16,kkp:16a}. This showed that the Insider framework can also encompass other methodologies, here,
the inductive approach to protocol verification \cite{pau:98}. 
Unintentional insiders in social networking are a further application \cite{kam:22b} as is the use of the
framework for the use of black box AI \cite{kam:22d}.
Airplane case studies on insider threats \cite{kk:16,kk:21} exhibited further requirements for the framework:
dynamic state verification was added to the core formalisation thus formalizing a mutable state. At the same
time, Kripke structures and the temporal logic CTL have been formalized supporting Modelchecking. Also a
semantics for attack trees has been formalized on this basis and proved to be correct and complete
\cite{kam:17a,kam:17b,kam:17c,kam:18b,kam:19a}.
Attack trees enabled integrating Isabelle formal reasoning for IoT systems in the CHIST-ERA project SUCCESS \cite{suc:16}.
There, attack trees combined with the Behaviour Interaction Priority (BIP) component architecture model enabled
developing security and privacy enhanced IoT solutions.
Most of the applications of the framework required the formalisation of technical and physical aspects of sysmte
and less so psychological features of the actors.
Thus, the development of the Isabelle {\it Insider} framework was rather an Isabelle {\it Infrastructure} framework.
The name of the framework has then be established as the all encompassing Isabelle Infarstructure and Insider framework
(IIIf)
The strong expressiveness of Isabelle allows to formalize the IoT scenarios as well as actors and policies.
the IIIf has also been applied to evaluate IoT scenarios with respect to policies like the European data privacy
regulation GDPR \cite{kam:18a}.
The security protocol application domain, initiated by auction protocols \cite{kkp:16,kkp:16a} has led to explore
the analysis of Quantum Cryptography which in turn necessitated the extension by probabilities \cite{kam:19b,kam:19c,kam:19d}.
Recently the IIIf has been used to show how refinement can develop practical implementations of the
Quantum Key Distribution (QKD) protocol. These implementations aim at 
overcoming Phonton-Number-Splitting attacks by the extension of QKD to Decoy QKD. The IIIf refinement
could be successfully applied to prove absence of these attacks \cite{knpw:25}.

Requirements raised by these various security and privacy case studies have shown the need for a
cyclic engineering process for developing specifications and refining them towards implementations.
Traditionally, formal specifications can be derived from semi-formal (UML) system specifications in
a rigorous way \cite{bk:03} leading to good starting points for this enginnering process.
A first case study takes the IoT healthcare application and exemplifies a step-by-step
refinement interspersed with attack analysis using attack trees to increase privacy by ultimately
introducing a blockchain for access control \cite{kam:19a}.
Support for a dedicated security refinement process (Refinement Risk cycle, RR cycle) are proposed
and experimented with in the paper \cite{kam:20a}. The publication \cite{kl:20} illustrated security
refinement on the development of a decentralized privacy preserving architecture for the Corona Warning
App (CWA) developing and applying a formalized notion of a refinement process, This is then also validated
in a reconsideration of a more systematic introduction of a blockchain into the IoT healthcare study \cite{kam:23b}.

The process of security engineering is additionally extended to information flow security in \cite{kam:24a} by
exploring the relationship of NI and refinement via the notion of a shadow. Another experimentation with the
relationship of distributed GIS systems and machine learning lead to extending IIIf by a notion of Federated
Learning for differentially privacy \cite{kam:24b}. To this end, a distributed version of NI based on probabilities
is formalized vaguely inspired by Volpano and Smith's \cite{vs:99} intuition of how to express probabilistic NI for
distributed systems.
This work has now been extended to relate to the classical notion of attacker advantage.

\subsection{Related Work on DP, FL, and NI}
Information theoretic aspects of DP are studied intensively at the moment as is
documented in the survey paper \cite{am:23}. While the foundations appear to become clarified, the relationship
to practical implementations of DP, that is, inference or verification systems and thus notions like
noninterference are still fairly open. In particular, works that support the formalization and automated
reasoning are very scarce with a few notable exceptions. To our knowledge, amongst this there are none in
relation to federated learning and modeling distribution of data, communication and processing as we do using the
infrastructures, actors, and actions of the IIIf.
Tschantz et al \cite{tkd:11} already observed the similarity between noninterference and differential privacy
with respect to verification. More recently, Liu et al \cite{lwfz:23} took these concepts further and presented a
dedicated model checking technique.
However, both works do not consider the conceptual relation of FL and DP to fundamental notions of Cryptography
like computational hardness and $\text{Adv}(\Att)$.

\subsection{Conclusion}
The notion of DP naturally corresponds to the NI property which gives rise to a natural way of defining
a quantified notion of NI-FL-DP (see Section \ref{sec:nifldp}). The quantification is useful to relate
practically motivated attack techniques, like the DLG, and their defense techniques, that is, gradient
sparsification and pseudo-gradients or fundamental notions of cryptography and to rigorous formalisation
and automated reasoning and verification in IIIf.

\bibliographystyle{abbrv}
\bibliography{biblio}

\end{document}